\documentclass[11pt]{article}
\usepackage{geometry}
\usepackage{hyperref}

\usepackage{amsmath}

\usepackage{fork-ip-macros}

\patchcmd{\maketitle}{\@fnsymbol}{\@arabic}{}{}  
\makeatother

\title{Simulation Theorems via Pseudo-random Properties}
\author{
  Arkadev Chattopadhyay \thanks{Tata Institute of Fundamental Research, Mumbai, \texttt{arkadev.c@tifr.res.in}}
  \and
   Michal Kouck\'y\thanks{Charles University, Prague, \texttt{koucky@iuuk.mff.cuni.cz}}
   \and
  Bruno Loff\thanks{INESC-TEC and University of Porto, Porto, \texttt{bruno.loff@gmail.com}}
  \and
   Sagnik Mukhopadhyay \thanks{Tata Institute of Fundamental Research, Mumbai, \texttt{sagnik@tifr.res.in}}} 
\date{}

\begin{document}

\maketitle

\setcounter{tocdepth}{2}

\abstract{

We generalize the deterministic simulation theorem of Raz and McKenzie \cite{RM99}, to any gadget which satisfies certain hitting property. We prove that inner-product and gap-Hamming satisfy this property, and as a corollary we obtain deterministic simulation theorem for these gadgets, where the gadget's input-size is logarithmic in the input-size of the outer function. This answers an open question posed by G\"{o}\"{o}s, Pitassi and Watson \cite{GPW15}.
Our result also implies the previous results for the Indexing gadget, with better parameters than was previously known. A preliminary version of the results obtained in this work appeared in \cite{CKL+17}.


\thispagestyle{empty}


\bigskip\bigskip\bigskip\bigskip
\begin{center}
\tableofcontents
\end{center}

 \newpage

\setcounter{page}{1}

\section{Introduction} \label{SEC:INTRO}

 A very basic problem in computational complexity is to understand the \emph{complexity} of a composed function $f \circ g$ in terms of the complexities of the two simpler functions $f$ and $g$ used for the composition. For concreteness, we consider $f:\{0,1\}^p \to \cZ$ and $g:\{0,1\}^m \to \{0,1\}$ and denote the composed function as $f \circ g^p: \ZO^{m p} \to \cZ$; then $f$ is called the \textit{outer-function} and $g$ is called the \textit{inner-function}. The special case of $\cZ$ being $\{0,1\}$ and $f$ the $\XOR$ function has been the focus of several works \cite{Yao82, Lev87, Imp95,Sha03,LS08,VW08, She12}, commonly known as XOR lemmas. Another special case is when $f$ is the trivial function that maps each point to itself. This case has also been widely studied in various parts of complexity theory under the names of `direct sum' and `direct product' problems, depending on the quality of the desired solution \cite{JRS03, BPSW05, HJMR07, JKN08, Dru12, Pan12, JPY12, JY12, BBCR13, BRWY13a, BRWY13, BBK+13, BR14, KLL+15, Jai15}. Making progress on even these special cases of the general problem in various models of computation is an outstanding open problem.

 While no such general theorems are known, there has been some progress in the setting of communication complexity. In this setting the input for $g$ is split between two parties, Alice and Bob.
 A particular instance of progress from a few years ago is the development of the pattern matrix method by Sherstov \cite{She11} and the closely related block-composition method of Shi and Zhu \cite{SZ09}, which led to a series of interesting developments \cite{Cha07,LS08,CA08,She12b, She13, RY15}, resolving several open problems along the way. In both these methods, the relevant analytic property of the outer function is approximate degree. While the pattern-matrix method entailed the use of a special inner function, the block-composition method, further developed by Chattopadhyay \cite{C09}, Lee and Zhang \cite{LZ10} and Sherstov \cite{She12b, She13}, prescribed the inner function to have small discrepancy. These methods are able to lower bound the randomized communication complexity of $f \circ g^p$ essentially by the product of the approximate degree of $f$ and the logarithm of the inverse of discrepancy of $g$. 

 The following simple protocol is suggestive: Alice and Bob try to solve $f$ using a decision tree (randomized/deterministic) algorithm. Such an algorithm queries the input bits of $f$ frugally. Whenever there is a query, Alice and Bob solve the relevant instance of $g$ by using the best protocol for $g$. This allows them to progress with the decision tree computation of $f$, yielding (informally) an upper bound of ${\mathcal M}^{cc} \big(f \circ g^p\big) = O({\mathcal M}^{dt}\big(f\big)\cdot {\mathcal M}^{cc}\big(g\big))$, where $\mathcal M$ could be the deterministic or randomized model and ${\mathcal M}^{dt}$ denotes the decision tree complexity. A natural question is if the above upper bound is essentially optimal. The case when both $f$ and $g$ are $\XOR$ clearly shows that this is not always the case. However, this may just be a pathological case. Indeed it is natural to study for what models $\cM$ and which inner functions $g$, is the above naive algorithm optimal.

In a remarkable and celebrated work, Raz and McKenzie \cite{RM99} showed that this na\"{\i}ve upper bound is always optimal for \emph{deterministic protocols}, when $g$ is the Indexing function ($\IND$), provided the \emph{gadget size is polynomially large} in $p$. This theorem was the main technical workhorse of Raz and McKenzie to famously separate the monotone NC hierarchy. The work of Raz and McKenzie was recently simplified and built upon by G\"o\"os, Pitassi and Watson \cite{GPW15} to solve a longstanding open problem in communication complexity. In line with \cite{GPW15}, we call such theorems \emph{simulation theorems}, because they explicitly construct a decision-tree for $f$ by simulating a given protocol for $f\circ g^p$. More recently, de Rezende, Nordstr\"om and Vinyals \cite{dNV16} port the above deterministic simulation theorem to the model of real communication, yielding new trade-offs for the measures of size and space in the cutting planes proof system.
 
 \bsni
  In this work, our main result is the following:
 \begin{theorem}
    \label{thm:det-simulation}
  Let $p \le 2^{\frac{n}{200}}$, $f:\{0,1\}^p \to \cZ$, where $\cZ$ is any domain, and $g: \zon \times \zon \rightarrow \{0,1\}$ be inner-product function, or any function from the gap-Hamming class of promise-problems. Then, 
 $$\cD^{cc}\big(f \circ g^p\big) = \Theta\bigg(\cD^{dt}\big(f\big)\cdot n\bigg).$$ 
 \end{theorem}

 The inner-product function $\IP_n \zon \times \zon \rightarrow \ZO$ is defined as $\IP_n(x,y) = \sum_{i \in [n]} x_i \cdot y_i$, where the summation is taken over field $\mathbb{F}_2$. Problems in the class of gap-Hamming promise-problems, parameterized with $\gamma$ and denoted by$\GH_{n,\gamma}: \zon \times \zon \rightarrow \ZO$, distinguish the case of $(x,y)$ having Hamming distance at least $(\frac{1}{2}+\gamma) n$ from the case of $(x,y)$ having Hamming distance at most $(\frac{1}{2}-\gamma)n$, for $0 \le \gamma \le 1/4$. Note that this is the first deterministic simulation theorem with logarithmic gadget size, whereas the Raz-McKenzie simulation theorem requires a polynomial size gadget. This answers a problem raised by both G\"o\"os-Pittasi-Watson \cite{GPW15} and G\"o\"os et.al. \cite{GLM+15} of proving a Raz-McKenzie style deterministic simulation theorem for a different inner function than Indexing with a better gadget size. Moreover, it is not hard to verify that an instance  of the function $g$ easily embeds in Indexing by exponentially blowing up the size. This enables us to also re-derive the original Raz-McKenzie simulation theorem for the Indexing function, even attaining significantly better parameters. This improvement in parameters answers a question posed to us recently by Jakob Nordstr\"om \cite{Nordstrom16}.

 The techniques required to prove the deterministic simulation theorem are based on those which appear in \cite{RM99,GPW15}. Our contribution in this part is two-fold. On one hand, we generalize the proof considerably, by singling-out a new pseudo-random property of a function $g:\ZO^n\to\ZO$, which we call ``having $(\delta,h)$-hitting rectangle-distributions'', and then showing that a simulation theorem will hold ($\cD^{cc}(f_p \circ g^p) = \Theta(\cD^{dt}(f)\cdot h)$) for any $g$ with this property. We then show that the inner-product function and the gap-Hamming problem have the above property. This results in a simulation theorem for $\IP$ and $\GH$ with exponentially smaller gadget size than was previously known. We discuss the pseudo-random property and its connection to gadget-size in the next sub-section.
 
 It is well known that inner-product has strong pseudo-random properties. In particular it has vanishing discrepancy under the uniform distribution which makes it a good 2-source extractor. In fact, such strong properties of inner-product were recently used to prove simulation theorems for more exotic models of communication by G\"o\"os et al. \cite{GLM+15} and also by the authors and Dvo\v{r}\'{a}k \cite{CDK+17}  to resolve a problem with a direct-sum flavor. By comparison, the pseudo-random property we abstract for proving our simulation theorem seems milder. This intuition is corroborated by the fact that we can show that gap-Hamming problems also possess our property, even though we know that these problems have large $\Omega(1)$ discrepancy under all distributions. Interestingly, any technique that relies on the inner-function having small discrepancy, such as the block-composition method, will not succeed in proving simulation theorems for such inner gadgets.
 
 \medskip
We would, at this point, like to point out to the readers that a preliminary version of the results obtained in this paper appeared in \cite{CKL+17}.
 
 \bigskip
 We remark here that Wu, Yao and Yuen \cite{WYY17} have independently reported a proof of the simulation theorem for the inner-product function, while a draft of this manuscript was already in circulation. Implicit in their proof is the construction of hitting rectangle-distributions for $\IP$, and their construction of these distributions is similar to our own. This suggests that our pseudo-random property is essential to how simulation theorems are currently proven.

 \subsection{Our techniques} \label{SEC:TECHNIQUES}


The main tool for proving a tight deterministic simulation theorem is to use the general framework of the Raz-McKenzie theorem as expounded by G\"o\"os-Pittasi-Watson \cite{GPW15}. Given an input $z\in \ZO^p$ for $f$, and wishing to compute $f(z)$, we will query the bits of $z$ while simulating (in our head) the communication protocol for $f \circ g^p$, on inputs that are consistent with the queries to $z$ we have made thus far.
Namely, we maintain a rectangle $A \times B \subseteq \ZO^{n p} \times \ZO^{n p}$ so that for any $(x,y)\in A\times B$, $g^p(x,y)$ is \emph{consistent} with $z$ on all the coordinates that were queried. 
We will progress through the protocol with our rectangle $A\times B$ from the root to a leaf. As the protocol progresses, $A\times B$ shrinks
according to the protocol, and our goal is to maintain the consistency requirement.
For that we need that inputs in $A\times B$ allow for all possible answers of $g$ on those coordinates which we did not yet query.
Hence $A\times B$ needs to be rich enough, and we are choosing a path through the protocol that affects this richness the least.
If the protocol forces us to shrink the rectangle $A\times B$ so that we may not be able to maintain the richness condition,
we query another coordinate of $z$ to restore the richness. 
Once we reach a leaf of the protocol we learn a correct answer for $f(z)$, because there is an input $(x,y) \in A\times B$ on which $g^p(x,y)=z$ (since we preserved consistency) and all inputs in $A\times B$ give the same answer for $f \circ g^p$,

The technical property of $A\times B$ that we will maintain and which guarantees the necessary richness is called {\em thickness}. 
$A\times B$ is thick on the $i$-th coordinate if for each input pair $(x,y) \in A \times B$, even after one gets to see all the coordinates of $x$ and $y$ except for $x_i$ and $y_i$, the \emph{uncertainty} of what appears in the $i$th coordinate remains large enough so that $g(x_i, y_i)$ can be arbitrary. 
Let us denote by $\Ext^i_A(x_1,\dots, x_{i-1},x_{i+1},\dots, x_p)$ the set of possible extensions $x_i$ such that $\langle x_1,\dots,x_p\rangle \in A$. We define
$\Ext^i_B(y_1,\dots, y_{i-1},y_{i+1},\dots, y_p)$ similarly. If for a given $x_1,\dots, x_{i-1},x_{i+1},\dots, x_p$ and $y_1,\dots,\allowbreak y_{i-1},\allowbreak y_{i+1},\allowbreak \dots,y_p$
we know that both $\Ext^i_A(x_1,\dots, x_{i-1},x_{i+1},\dots, x_p)$ and $\Ext^i_B(y_1,\dots, y_{i-1},y_{i+1},\dots, y_p)$ are of size at least $2^{(\frac{1}{2}+\epsilon)n}$
then for $g=\IP_n$ there are extensions $x_i \in  \Ext^i_A(x_1,\dots, x_{i-1},x_{i+1},\dots, x_p)$ and $y_i \in \Ext^i_B(y_1,\dots, y_{i-1},y_{i+1},\dots, y_p)$ such that $\IP_n(x_i,y_i) = z_i$. Hence, we say that $A\times B$ is $\tau$-thick if $\Ext^i_A(x_1,\dots, x_{i-1},x_{i+1},\dots, x_p)$ and $\Ext^i_B(y_1,\dots, y_{i-1},y_{i+1},\dots, y_p)$ are of size at least $\tau \cdot 2^n$, for every choice of $i$ and $x_1,\dots,x_p \in A$, $y_1,\dots,y_p \in B$.

So if we can maintain the thickness of $A\times B$, we maintain the necessary richness of $A\times B$. It turns out that this is indeed possible using the
technique of Raz-McKenzie and G\"o\"os-Pittasi-Watson. Hence as we progress through the protocol we maintain $A\times B$ to be $\tau$-thick and dense.
Once the density of either $A$ or $B$ drops below certain level we are forced to make a query to another coordinate of $z$. Magically, that restores
the density (and thus thickness) of $A\times B$ on coordinates not queried. (An intuitive reason is that if the density of extensions in some coordinate is low
then the density in the remaining coordinates must be large.)

We capture the property of the inner function $g$ that allows this type of argument to work, as follows.
For $\delta \in (0,1)$ and integer $h\ge 1$ we say that $g$ has {\em $(\delta,h)$-hitting monochromatic rectangle distributions} if
there are two distributions $\sigma_0$ and $\sigma_1$ where for each $c \in \bool$,  $\sigma_c$ is a distribution over $c$-monochromatic rectangles $U \times V \subset \bool^{n} \times \bool^{n}$ (i.e., $g(u, v) = c$ on every pair $(u, v) \in U \times V$), such that for any set $X \times Y \subset \bool^{n} \times \bool^{n}$ of sufficient size, a rectangle randomly chosen according to $\sigma_c$ will intersect $X \times Y$ with large probability. More precisely, for any $c \in \bool$ and for any $X \times Y$ with $|X|/2^{n}, |Y|/2^{n} \geq 2^{-h}$,
\begin{align*}
\Pr_{(U \times V) \sim \sigma_c}[(U \times V) \cap (X \times Y) \neq \varnothing] \geq 1 -  \delta.
\end{align*}
If such distributions $\sigma_0$ and $\sigma_1$ exist, we say that $g$ has $(\delta, h)$-hitting monochromatic rectangle-distributions. We then prove the following:

\begin{theorem}\label{thm:det-gen-simulation} 
If $g$ has $(\delta, h)$-hitting monochromatic rectangle-distributions, $\delta < 1/6$, and $p \le 2^{\frac{h}{2}}$, then 
\[ 
\cD^{dt}(f) \le \frac{8}{h} \cdot \cD^{cc}(f \circ g^{\,p}). 
\] 
\end{theorem}

We prove this general theorem and then establish that $\GH$ and $\IP$ over $n$-bits has $(o(1), \Omega(n))$-hitting rectangle-distributions. This immediately yields Theorem~\ref{thm:det-simulation}.

The distribution $\sigma_0$ for $\GH_{n,\frac{1}{4}}$ is sampled as follows: we first sample a random string $x$ of Hamming weight $\frac{n}{2}$, and we look at the set of all strings of Hamming weight $\frac{n}{2}$ which are at Hamming distance at most $\frac{n}{8}$ from $x$. Let's call this set $U_x$. The output of $\sigma_0$ will be the rectangle $U_x \times U_x$. The output of $\sigma_1$ is $U_x \times U_{\bar{x}}$, where $\bar{x}$ is the bit-wise complement of $x$. For any such $x$, $U_x \times U_x$ will be a 0-monochromatic rectangle and $U_x \times U_{\bar{x}}$ will be a 1-monochromatic rectangle. Note that if $U_x$ does not hit a subset $A$ of $\zon$, then it means that $x$ is at least $\frac{n}{8}$ Hamming distance away from the set $A$. By an application of Harper's theorem, we can show that for a sufficiently large set $A$, the number of strings which are at least $\frac{n}{8}$ Hamming distance away from $A$ is exponentially small. This will imply that both $\sigma_0$ and $\sigma_1$ will hit a sufficiently large rectangle with probability exponentially close to $1$, which is our required hitting property.

The $\sigma_0$ distribution for $\IP_n$ is picked as follows: To produce a rectangle $U\times V$ we sample uniformly at random a linear sub-space 
$V \subseteq F_2^n$ of dimension $n/2$ and we set $U=V^{\perp}$ to be the orthogonal complement of $V$. Since a random vector space of size $2^{n/2}$ hits
a fixed subset of $\ZO^n$ of size $2^{(\frac{1}{2} + \epsilon) n}$ with probability  $1-O(2^{-\epsilon n})$, and both $U$ and $V$ are random vector spaces
of that size, $U\times V$  intersects a given rectangle $X\times Y$ with probability  $1-O(2^{-\epsilon n})$. Hence, we obtain $(O(2^{-\epsilon n}), (\frac{1}{2} + \epsilon) n)$-hitting distribution for $\IP$. For the $1$-monochromatic case, we first pick a random $a \in F_2^n$ of odd hamming weight
and them pick random $V$ and $U=V^{\perp}$ inside of the orthogonal complement of $a$. The distribution $\sigma_1$ outputs the $1$-monochromatic rectangle $(a+V) \times (a+U)$, and will have the required hitting property.

\subsection{Organization}
Section \ref{SEC:DEF-PRELIM} consists of basic definitions and preliminaries. In Section \ref{SEC:DET} we prove a deterministic simulation theorem for any gadget admitting $(\delta, h)$-hitting monochromatic rectangle-distribution: sub-section \ref{sec:thick} provides some supporting lemmas for the proof, and sub-section \ref{SEC:G-DET} holds the proof itself. In Section \ref{SEC:RGH} we show that $\GH_{n, \frac{1}{4}}$ on $n$-bits has $(o(1), \frac{n}{100})$-hitting rectangle distribution, and in Section \ref{SEC:IP-DET} we show that $\IP$ on $n$-bits has $(o(1), n/5)$-hitting rectangle distribution.

\section{Basic definitions and preliminaries} \label{SEC:DEF-PRELIM}

A \emph{combinatorial rectangle}, or just a \emph{rectangle} for short, is any product $A\times B$, where both $A$ and $B$ are finite sets. If $A' \subseteq A$ and $B' \subseteq B$, then $A'\times B'$ is called a \emph{sub-rectangle} of $A\times B$. The \emph{density} of $A'$ in $A$ is $\alpha=|A'|/|A|$.
\medskip

Consider a product set $\cA = \cA_1 \times \ldots \times \cA_p$, for some natural number $p \ge 1$, where each $\cA_i$ is a subset of $\ZO^n$. Let $A \subseteq \cA$ and $I \subseteq [p] \eqdef \{1, \ldots, p\}$. Let $I=\{i_1<i_2 < \cdot <i_k\}$, and $J=[p]\setminus I$. For any $a \in (\ZO^n)^p$, we let $a_I = \langle a_{i_1}, a_{i_2}, \dots, a_{i_k}\rangle$ be the projection of $a$ onto the coordinates in $I$. Correspondingly, $A_I = \{ a_I \mid a \in A \}$ is the projection of the entire set $A$ onto $I$. For any $a' \in (\ZO^n)^{k}$ and $a'' \in (\ZO^n)^{p - k}$, we denote by $a' \times_I a''$ the $p$-tuple $a$ such that $a_I = a'$ and $a_{J} = a''$. If $I = [k]$ for some $k \le p$, we may omit the set $I$ and write only $a' \times a''$. For $i\in [p]$ and a $p$-tuple $a$, $a_{\neqi}$ denotes $a_{[p]\setminus \{i\}}$, and similarly, $A_{\neqi}$ denotes $A_{[p]\setminus \{i\}}$. For $a' \in (\ZO^n)^{k}$, we define the set of extensions $\Ext_A^{J}(a') = \{a'' \in (\ZO^n)^{p - k} \mid a' \times_I a'' \in A\}$; we call those $a''$ \emph{extensions} of $a'$. Again, if $A$ and $I$ are clear from the context, we may omit them and write only $\Ext(a')$.

Suppose $n\ge 1$ is an integer and $\cA=\ZO^n$.
For an integer $p$, a set $A \subseteq \cA^p$ and a subset $S \subseteq \cA$, the restriction of $A$ to $S$ at coordinate $i$ is the set $A^{i,S} = \{a \in A \mid a_i \in S\}$. We write $A^{i,S}_I$ for the set $(A^{i, S})_I$ (i.e. we first restrict the $i$-th coordinate then project onto the coordinates in $I$). Clearly $A^{i, S}_{\neqi}$ is non-empty if and only if $S$ and $A_i$ intersect.

The density of a set $A \subseteq \cA^p$ will be denoted by $\alpha = \frac{|A|}{|\cA|^p}$, and $\alpha^{i,S}_I = \frac{|A^{i,S}_I|}{|\cA|^{|I|}}$.

\subsection*{Communication complexity}

See \cite{KN97} for an excellent exposition on this topic, which we cover here only very briefly. In the two-party communication model introduced by Yao \cite{Yao79}, two computationally unbounded players, Alice and Bob, are required to jointly compute a function $F: \cA \times \cB \rightarrow \cZ$ where Alice is given $a \in \cA$ and Bob is given $b \in \cB$. To compute $F$, Alice and Bob communicate messages to each other, and they are charged for the total number of bits exchanged.

Formally, a \emph{deterministic protocol} $\pi:\cA\times\cB \to \cZ$ is a binary tree where each internal node $v$ is associated with one of the players; Alice's nodes are labeled by a function $a_v: \cA \rightarrow \bool$, and Bob's nodes by $b_v : \cB \rightarrow \bool$. Each leaf node is labeled by an element of $\cZ$. For each internal node $v$, the two outgoing edges are labeled by 0 and 1 respectively. 
The \emph{execution} of $\pi$ on the input $(a,b) \in \cA\times\cB$ follows a path in this tree: starting from the root, in each internal node $v$ belonging to Alice, she communicates $a_v(a)$, which advances the execution to the corresponding child of $v$; Bob does likewise on his nodes, and once the path reaches a leaf node, this node's label is the output of the execution. We say that $\pi$ \emph{correctly computes} $F$ on $(a, b)$ if this label equals $F(a, b)$.

\medskip

To each node $v$ of a deterministic protocol $\pi$ we associate a set $R_v \subseteq \cA \times \cB$ comprising those inputs $(a,b)$ which cause $\pi$ to reach node $v$. It is easy see that this set $R_v$ is a combinatorial rectangle, i.e. $R_v = A_v \times B_v$ for some $A_v \subseteq \cA$ and $B_v \subseteq \cB$. 

\medskip

The \emph{communication complexity of $\pi$} is the height of the tree. The \emph{deterministic communication complexity of $F$}, denoted $\cD^{cc}(F)$, is defined as the smallest communication complexity of any deterministic protocol which correctly computes $F$ on every input.

\subsection*{Decision tree complexity}

In the (Boolean) decision-tree model, we wish to compute a function $f:\ZO^p\to\cZ$ when given query access to the input, and are charged for the total number of queries we make.

Formally, a \emph{deterministic decision-tree} $T:\ZO^p \to \cZ$ is a rooted binary tree where each internal node $v$ is labeled with a variable-number $i \in [p]$, each edge is labeled $0$ or $1$, and and each leaf is labeled with an element of $\cZ$. The execution of $T$ on an input $z \in \ZO^p$ traces a path in this tree: at each internal node $v$ it queries the corresponding coordinate $z_i$, and follows the edge labeled $z_i$. Whenever the algorithm reaches a leaf, it outputs the associated label and terminates. We say that $T$ \emph{correctly computes} $f$ on $z$ if this label equals $f(z)$.

The \emph{query complexity} of $T$ is the height of the tree. The \emph{deterministic query complexity} of $f$, denoted $\cD^{dt}(F)$, is defined as the smallest query complexity of any deterministic decision-tree which correctly computes $f$ on every input.

\subsection*{Functions of interest}

The \emph{Inner-product function on $n$-bits}, denoted $\IP_n$ is defined on $\bool^n \times \bool^n$ to be:
\begin{align*}
\IP_n(x,y) = \sum_{i \in [n]} x_i\cdot y_i \mod 2.
\end{align*}

\bsni
For $N=2^n$, the \emph{Indexing function on $N$-bits}, $\IND_N$, is defined on $\bool^{\log N} \times \bool^{N}$ to be:
\begin{align*}
\IND_N(x,y) = y_x \quad \text{(the $x$'th bit of $y$).}
\end{align*}

\bsni
Let $n$ be a natural number and $\gamma = \frac{k}{n} \in (0, 1/2)$. For two $n$-bit strings $x$ and $y$, let $d_H(x, y) = \sum_i x_i \oplus y_i$ be their Hamming-distance. The \emph{gap-Hamming problem}, denoted $\GH_{n, \gamma}$ is a promise-problem defined on $\ZO^n\times\ZO^n$, by the condition
\[
\GH_{n, \gamma}(x, y) = \begin{cases}
1 & \tif d_H(x,y) \ge (\frac 1 2 + \gamma) \; n,\\
0 & \tif d_H(x, y) \le (\frac 1 2 - \gamma) \; n.
\end{cases}
\]


\section{Deterministic simulation theorem}
 \label{SEC:DET}
 
A \emph{simulation theorem} shows how to construct a decision tree for a function $f$ from a communication protocol for a composition problem $f \circ g^{p}$. Such a theorem can also be called a \emph{lifting} theorem, if one wishes to emphasize that lower-bounds for the decision-tree complexity of $f$ can be \emph{lifted} to lower-bounds for the communication complexity of $f\circ g^p$.  As mentioned in Section \ref{SEC:INTRO}, the deterministic lifting theorem proved in \cite{RM99}, and subsequently simplified in \cite{GPW15}, uses $\IND_N$ as inner function $g$ with $N$ being polynomially larger than $p$. In this section we will show a deterministic simulation theorem for any function which possesses a certain pseudo-random property, which we will now define. Later we will show that the Inner-product function has this property.
 
\begin{definition}[Hitting rectangle-distributions] 
Let $0\le \delta < 1$ be a real, $h\ge 1$ be an integer, and $\cA,\cB$ be some sets.
   A distribution $\sigma$ over rectangles within $\cA\times\cB$ is called a \emph{$(\delta, h)$-hitting rectangle-distribution} if, for any rectangle $A\times B$ with $|A|/|\cA|, |B|/|\cB| \ge 2^{-h}$, 
\[ 
  \Pr_{R \sim \sigma}[R \cap (A\times B) \neq \varnothing] \ge 1 - \delta. 
\]
\end{definition}

\bsni Let $g:\cA\times\cB \to \ZO$ be a (possibly partial) function. A rectangle $A\times B$ is $c$-monochromatic with respect to $g$ if $g(a, b) = c$ for every $(a,b)\in A\times B$.

\begin{definition} 
For a real $\delta\ge 0$ and an integer $h\ge 1$, we say that a (possibly partial) function \emph{$g:\cA\times\cB \to \ZO$ has $(\delta, h)$-hitting monochromatic rectangle-distributions} if there are two $(\delta, h)$-hitting rectangle-distributions $\sigma_0$ and $\sigma_1$, where each $\sigma_c$ is a distribution over rectangles within $\cA\times\cB$ that are $c$-monochromatic with respect to $g$.
\end{definition} 
 
\bsni 
The theorem we will prove in Section \ref{SEC:G-DET} is the following: 
 
\begin{theorem}\label{thm:simulation} Let $\eps \in (0, 1)$ and $\delta \in (0,\frac 1 6)$ be real numbers, and let $h\ge 6/\eps$ and $1 \le p \le 2^{h (1 - \eps)}$ be integers. Let $f:\ZO^p \rightarrow \cZ$ be a function and $g: \cA \times \cB \rightarrow \ZO$ be a (possibly partial) function. If $g$ has $(\delta, h)$-hitting monochromatic rectangle-distributions then
\[ 
\cD^{dt}(f) \le \frac{4}{\eps \cdot h} \cdot \cD^{cc}(f \circ g^{\,p}). 
\] 
\end{theorem}

\bsni

\bsni 
In Section \ref{SEC:RGH} we will show that $\GH_{n,\frac{1}{4}}$ has $(o(1), \frac{n}{100})$-hitting monochromatic rectangle-distributions. From this we obtain a simulation theorem for $\GH_{n, \frac{1}{4}}$:
 
\begin{corollary} Let $n$ be large enough even integer, and $p \le 2^{\frac{n}{200}}$ be an integer. For any function $f:\ZO^p \rightarrow \ZO$, $\cD^{dt}(f) \le O(\frac{1}{n} \cdot \cD^{cc}(f \circ \GH^{\,p}_{n,\frac{1}{4}}))$.
\end{corollary}

\bsni
In Section \ref{SEC:IP-DET} we will show that $\IP_n$ has $(o(1), n (\frac{1}{2} - \eps))$-hitting monochromatic rectangle-distributions, for any constant $\eps \in (0, 1/2)$. This allows us to derive:
 
\begin{corollary} Let $n$ be large enough integer, $\eps \in (0, 1/2)$ be a constant real, and $p \le 2^{(\frac{1}{2} - \eps) n}$ be an integer. For any function $f:\ZO^p \rightarrow \ZO$, $\cD^{dt}(f) \le \frac{10}{n \eps} \cdot \cD^{cc}(f \circ \IP^{\,p}_n)$.
\end{corollary}

\bsni
This allow us to significantly improve the gadget size known for simulation theorem of \cite{RM99, GPW15}, that uses the Indexing function instead of Inner-Product. Indeed, Jakob Nordstr\"om \cite{Nordstrom16} recently posed to us the challenge of proving a simulation theorem for $f \circ \IND_N^p$, with a gadget size $N$ smaller than $p^3$ ($p^3$ is already a significant improvement to \cite{RM99, GPW15}).

This follows from the above corollary, because of the following reduction: Given an instance $(a, b) \subseteq (\zonp)^2$ of $f \circ \IP_n^p$ where $p \leq 2^{n(\frac{1}{2} + \eps)}$, Alice and Bob can construct an instance of $f \circ \IND_N^p$ where $N = 2^n$. Bob converts his input $b \in \bool^{np}$ to $b' \in \bool^{N p}$, so that each $b'_i = [\IP_n(x_1, b_i) \rangle, \cdots , \IP_n(x_N, b_i) \rangle]$ where $\{x_1, \cdots, x_N\} = \bool^n$ is an ordering of all $n$-bit strings. It is easy to see that  $\IP_n(a_i,b_i) = \IND_N(a_i, b'_i)$. Hence it follows as a corollary to our result for $\IP$:

\begin{corollary}
	\label{COR:GPW}
Let $\eps \in (0, 1/2)$ be a constant real number, and $N$ and $p$ be sufficiently large natural numbers, such that $p^{2 + \eps} \le N$. Then $\cD^{dt}(f) = O( \frac{1}{\eps \cdot \log N} \cdot \cD^{cc}(f \circ \IND_N^{\,p}))$. 
\end{corollary}

 Also, it is worth noting that the proof of Lemma 7 (projection lemma) in \cite{GPW15} implicitly proves that $\IND_n$ has $(o(1), \frac 3 {20} \log N)$-hitting rectangle-distribution. Hence we can also apply Theorem \ref{thm:simulation} directly to obtain a corollary similar to Corollary \ref{COR:GPW} (albeit with much larger gadget size $N$).

 

\subsection{Thickness and its properties} \label{sec:thick}

\begin{definition}[Aux graph, average and min-degrees] Let $p \ge 2$. For $i \in [p]$ and $A \subseteq \cA^p$, the aux graph $G(A, i)$ is the bipartite graph with left side vertices $A_i$, right side vertices $A_{\neqi}$ and edges corresponding to the set $A$, i.e., $(a',a'')$ is an edge iff $a' \times_{\{i\}} a'' \in A$.

\bsni
We define the average degree of $G(A, i)$ to be the average right-degree:
\begin{align*}
\davg(A, i) = \frac{|A|}{|A_{\neqi}|},
\end{align*}
and the min-degree of $G(A, i)$, to be the minimum right-degree:
\begin{align*}
\dmin(A, i) = \min_{a' \in A_{\neqi}} |\Ext(a')|.
\end{align*}
\end{definition}

\begin{definition}[Thickness and average-thickness]
	\label{def:thickness}
 For $p\ge 2$ and $\tau, \varphi \in (0,1)$, a set $A \subseteq \cA^p$ is called {\em $\tau$-thick} if $\dmin(A, i) \geq \tau \cdot |\cA|$ for all $i \in [p]$.
(Note, an empty set $A$ is $\tau$-thick.)
 Similarly, $A$ is called {\em $\varphi$-average-thick} if $\davg(A, i) \geq \varphi\cdot |\cA|$ for all $i \in [p]$. For a rectangle $A\times B \subseteq \cA^p \times \cB^p$, we say that the rectangle $A \times B$ is $\tau$-thick if both $A$ and $B$ are $\tau$-thick. For $p=1$, set $A\subseteq \cA$ is $\tau$-thick if $|A| \ge \tau \cdot |\cA|$.
\end{definition}

\bsni
The following property is from \cite[Lemma 6]{GPW15}.

\begin{lemma}[Average-thickness implies thickness]
	\label{LEM:THICK}
For any $p\ge 2$, if $A \subseteq \cA^p$ is $\varphi$-average-thick, then for every $\delta \in (0, 1)$ there is a $\frac{\delta}{p} \varphi$-thick subset $A' \subseteq A$ with $|A'| \geq (1 - \delta) |A|$.
\end{lemma}

\begin{proof}
The set $A'$ is obtained by running Algorithm \ref{ALG:THICK}.

\begin{minipage}[H]{0.8\textwidth}
  \begin{algorithm}[H]
    \caption{}\label{ALG:THICK}
    \begin{algorithmic}[1]
      \small \State Set $A^0 = A$, $j=0$.  \While{$\dmin(A^j, i) < \frac{\delta}{p} \varphi \cdot 2^n$ for some $i \in [p]$} \State Let $a'$ be a right node of $G(A^j,i)$ with non-zero degree less than $\frac{\delta}{p} \varphi \cdot 2^n$.  \State Set $A^{j+1} = A^{j} \setminus \{ a' \} \times_{i} \Ext(a')$, i.e., remove every extension of $a'$. Increment $j$.
      \EndWhile
      \State Set $A'=A^j$.
    \end{algorithmic}
  \end{algorithm}
\end{minipage}

\bsni
The total number of iteration of the algorithm is at 
most $\sum_{i \in [p]} |A_{\neqi}|$. (We remove at least one node in some $G(A^j,i)$ in each iteration which was a node also in the original $G(A,i)$.) 
So the number of iterations is at most
\begin{align*}
\sum_{i \in [p]} |A_{\neqi}| = \sum_{i \in [p]} \frac{|A|}{\davg(A, i)} \leq \frac{p |A|}{\varphi 2^n}.
\end{align*}
As the algorithm removes at most $\frac{\delta}{p} \varphi \cdot 2^n$ elements of $A$ in each iteration, the total number of elements removed from $A$ is  
at most $\delta |A|$, so $|A'| \ge (1 - \delta) |A|$. Hence, the algorithm always terminates with a non-empty set $A'$ that must be $\frac{\delta}{p} \varphi$-thick. 
\end{proof}

\medskip

\begin{lemma}
  \label{LEM:PROJ-THICKNESS} Let $p\ge 2$ be an integer, $i\in [p]$, $A \subseteq \cA^p$ be a $\tau$-thick set, and $S \subseteq \cA$. The set $A^{i,S}_{\neqi}$ is $\tau$-thick. $A^{i,S}_{\neqi}$ is empty iff $S \cap A_i$ is empty. 
\end{lemma}

\begin{proof}
Notice that $A^{i, S}_{\neqi}$ is non-empty iff $S \cap A_i$ is non-empty.
Consider the case of $p\ge 3$.
Let $a\in A$, where $a_i \in S$. Set $a'=a_{\neqi}$. For $j' \in [p-1]$, let $j=j'+1$ if $j' \ge i$, and $j=j'$ otherwise. Clearly, $\Ext^{\{j\}}_A (a_{\neq j}) \subseteq \Ext^{\{j'\}}_{A^{i,S}_{\neqi}}(a'_{\neq j'})$, hence the degree of $a'$ in $G(A^{i,S}_{\neqi},j')$ is at least the degree of $a$
in $G(A,j)$ which is at least $\tau \cdot |\cA|$. Hence, $A^{i,S}_{\neqi}$ is $\tau$-thick.

To see the case $p = 2$, assume there is some string $a' \in A_{\neqi}$ which has some extension $a'' \in S$; but $A$ itself is $\tau$-thick, so there have to be at least $\tau \cdot |\cA|$ many such $a'$, which will then all be in $A^{i, S}_{\neqi}$.
\end{proof}



\begin{lemma}
	\label{LEM:PROJ}
Let $h\ge 1$, $p\ge 2$ and $i\in [p]$ be integers and $\delta,\tau,\varphi \in (0,1)$ be reals, where $\tau \geq 2^{-h}$.
Consider a function $g: \cA \times \cB \rightarrow \bool$ which has $(\delta, h)$-hitting monochromatic rectangle-distributions. Suppose $A \times B \subseteq \cA^p \times \cB^p$ is a non-empty rectangle which is $\tau$-thick, and suppose also that $\davg(A, i) \le \varphi \cdot |\cA|$. Then for any $c \in \bool$, there is a $c$-monochromatic rectangle $U \times V \subseteq \cA \times \cB$ such that
\begin{enumerate}
\item $A_{\neqi}^{i,U}$ and $B_{\neqi}^{i,V}$ is $\tau$-thick,
\item $\alpha_{\neqi}^{i,U} \geq \frac{1}{\varphi} (1 - 3\delta) \alpha$,
\item $\beta_{\neqi}^{i,V} \geq (1 - 3\delta) \beta$,
\end{enumerate}
where $\alpha = |A|/|\cA|^p$, $\beta = |B|/|\cB|^p$, $\alpha_{\neqi}^{i,U} = |A_{\neqi}^{i,U}|/|\cA|^{p-1}$ and $\beta = |B_{\neqi}^{i,U}|/|\cB|^{p-1}$.
\end{lemma}

{The constant $3$ in the statement may be replaced by any value greater than $2$, so the lemma is still meaningful for $\delta$ arbitrarily close to $1/2$.}

\begin{proof}
Fix $c\in \ZO$. Consider a matrix $M$ where rows correspond to strings $a \in A_{\neqi}$, and columns correspond to rectangles $R = U \times V$ in the support of $\sigma_c$. Set each entry $M(a, R)$ to $1$ if $U \cap \Ext^{\{i\}}_A(a) \neq \emptyset$, and set it to $0$ otherwise.

For each $a \in A_{\neqi}$, $|\Ext^{\{i\}}_A(a)| \ge \tau |\cA|$, and because $\sigma_c$ is a $(\delta, h)$-hitting rectangle-distribution and  $\tau \geq 2^{-h}$, we know that if we pick a column $R$ according to $\sigma_c$, then $M(a, R) = 1$ with probability $\ge 1 - \delta$. So the probability that $M(a, R) = 1$ over uniform $a$ and $\sigma_c$-chosen $R$ is $\ge 1 - \delta$.

Call a column of $M$ \emph{$A$-good} if $M(a, R) = 1$ for at least $1 - 3 \delta$ fraction of the rows $a$. Now it must be the case that the $A$-good columns have strictly more than ${1}/{2}$ of the $\sigma_c$-mass. Otherwise the probability that $M(a, R) = 1$ would be $< 1 - \delta$.

A similar argument also holds for Bob's set $B_{\neqi}$. Hence, there is a $c$-monochromatic rectangle $R = U\times V$ whose column is both $A$-good and $B$-good in their respective matrices. This is our desired rectangle $R$. 

We know: $|A_{\neqi}^{i,V}| \ge (1-3\delta) |A_{\neqi}|$ and $|B_{\neqi}^{i,V}| \ge (1-3\delta) |B_{\neqi}|$. Since $|B_{\neqi}| \geq |B|/|\cB|$,
we obtain $|B_{\neqi}^{i,V}|/|\cB|^{p-1} \ge (1-3\delta) |B_{\neqi}|/|\cB|^{p-1} \ge (1-3\delta) \beta $. 
Because $|A| / |A_{\neqi}| \le \varphi |\cA|$, we get
\[
\frac{|A_{\neqi}|}{|\cA|^{(p-1)}} \ge \frac{1}{\varphi}\cdot\frac{|A|}{|\cA|^{p}} = \frac{\alpha}{\varphi}.
\]
Combined with the lower bound on $|A_{\neqi}^{i,V}|$ we obtain $|A_{\neqi}^{i,U}|/|\cA|^{p-1} \ge (1-3\delta) \alpha / \varphi$.
The thickness of $A_{\neqi}^{i,U}$ and $B_{\neqi}^{i,V}$ follows from Lemma \ref{LEM:PROJ-THICKNESS}.
\end{proof}

\medskip

\begin{lemma}\label{LEM:REG} Let $p,h\ge 1$ be integers and $\delta,\tau \in (0,1)$ be reals, where $\tau \geq 2^{-h}$. 
Consider a function $g: \cA \times \cB \rightarrow \bool$ which has $(\delta, h)$-hitting monochromatic rectangle-distributions.
Let $A \times B \subseteq \cA^p\times\cB^p$ be a $\tau$-thick non-empty rectangle. Then for every $z \in \ZO^p$ there is some $(a, b) \in A\times B$ with $g^p(a, b) = z$.
\end{lemma}

\begin{proof}
This follows from repeated use of Lemma \ref{LEM:PROJ-THICKNESS}. Fix arbitrary $z \in \bool^p$. Set $A^{(1)}=A$ and $B^{(1)}=B$. We proceed in rounds 
$i = 1, \ldots, p-1$ maintaining a $\tau$-thick rectangle $A^{(i)} \times B^{(i)} \subseteq \cA^{p-i+1}\times\cB^{p-i+1}$.
If we pick $U_i \times V_i$ from $\sigma_{z_i}$, then the rectangle $(A^{(i)})_{\{i\}} \cap U_i \times (B^{(i)})_{\{i\}} \cap V_i$ will be non-empty with probability $\ge 1 - \delta > 0$ (because $\sigma_{z_i}$ is a $(\delta, h)$-hitting rectangle-distribution and $\tau \geq 2^{-h}$). Fix such $U_i$ and $V_i$. Set $a_i$ to an arbitrary string in $(A^{(i)})_{\{i\}} \cap U_i$, and $b_i$ to an arbitrary string in $(B^{(i)})_{\{i\}} \cap B_i$.
Set $A^{(i+1)} = (A^{(i)})^{i, \{a_i\}}_{\neqi}$, $B^{(i+1)} = (B^{(i)})^{i, \{b_i\}}_{\neqi}$, and proceed for the next round.
By Lemma~\ref{LEM:PROJ-THICKNESS}, $A^{(i+1)} \times B^{(i+1)}$ is $\tau$-thick.

Eventually, we are left with a rectangle $A^{(p)}\times B^{(p)} \subseteq \cA \times \cB$ where both $A^{(p)}$ and $B^{(p)}$ are $\tau$-thick (and non-empty). Again with probability $1 - \delta > 0$, the $z_p$-monochromatic rectangle $U_p\times V_p$ chosen from $\sigma_{z_p}$ will intersect $A^{(p)}\times B^{(p)}$. We again set $a_p$ and $b_p$ to come from the intersection, and set $a=\langle a_1,a_2,\dots,a_p \rangle$ and $b=\langle b_1,b_2,\dots,b_p \rangle$.
\end{proof}

\subsection{Proof of the simulation theorem}
\label{SEC:G-DET}

Now we are ready to present the simulation theorem (Theorem \ref{thm:simulation}). 
Let $\eps \in (0, 1/2)$ and $\delta \in (0,1/16)$ be real numbers, and $h\ge 6/\eps$ and $1 \le p \le 2^{h(1 - \eps)}$ be integers. Let $f:\ZO^p \rightarrow \cZ$ be a function and $g: \cA \times \cB \rightarrow \ZO$ be a (possibly partial) function. Assume that $g$ has $(\delta, h)$-hitting monochromatic rectangle-distributions.
We assume we have a communication protocol $\Pi$ for solving $f \circ g^p$, and we will use $\Pi$ to construct a decision tree (procedure) for $f$.
Let $C$ be the communication cost of the protocol $\Pi$.
If $p \le 5 C / h$ the theorem is true trivially. So assume $p > 5 C / h$. 
Set $\varphi = 4 \cdot 2^{-\eps h}$ and $\tau = 2^{-h}$. The decision-tree procedure is presented in Algorithm \ref{ALG:SIMUL}. On an input $z\in \ZO^p$, it uses the protocol $\Pi$ to decide which bits of $z$ to query.

The algorithm maintains a rectangle $A\times B \subseteq \cA^p \times \cB^p$ and a set $I \subseteq [p]$ of indices. $I$ corresponds
to coordinates of the input $z$ that were not queried, yet.

\renewcommand{\algorithmicrequire}{\textbf{Input:}}
\renewcommand{\algorithmicensure}{\textbf{Output:}}

\begin{minipage}[H]{0.8\textwidth}
\begin{algorithm}[H]
\caption{Decision-tree procedure}\label{ALG:SIMUL}
\begin{algorithmic}[1]
\small
\Require{$z \in \bool^p$}
\Ensure{$f(z)$}
\State Set $v$ to be the root of the protocol tree for $\Pi$, $I = [p]$, $A = \cA^p$ and $B = \cB^p$.
\While{$v$ is not a leaf} 
	\If{$A_I$ and $B_I$ are both $\varphi$-average-thick}
		\State Let $v_0, v_1$ be the children of $v$.
		\State Choose $c \in \ZO$ for which there is $A' \times B' \subseteq (A \times B) \cap R_{v_c}$ such that
		\State \ \ (1) $|A'_I \times B'_I| \geq \frac{1}{4} |A_I \times B_I|$
		\State \ \ (2) $A'_I \times B'_I$ is $\tau$-thick.
		\State Update $A= A'$, $B= B'$ and $v = v_c$.
	\ElsIf{$\davg(A_I, j) < \varphi |\cA|$ for some $j \in [|I|]$}
		\State Query $z_i$, where $i$ is the $j$-th (smallest) element of $I$.
		\State Let $U \times V$ be a $z_i$-monochromatic rectangle of $g$ such that 
		\State \ \ (1) $A_{I\setminus\{i\}}^{i,U} \times B_{I\setminus\{i\}}^{i,V}$ is $\tau$-thick,
		\State \ \ (2) $\alpha_{I\setminus\{i\}}^{i,U} \geq  \frac{1}{\varphi} (1 - 3 \delta) \alpha$,		
		\State \ \ (3) $\beta_{I\setminus\{i\}}^{i,V} \geq (1 - 3 \delta) \beta$,
		\State Update $A = A^{i,U}, B = B^{i,V}$ and $I = I \setminus \{ i \}$.
	\ElsIf{$\davg(B_I, j) < \varphi |\cB|$ for some $j \in [|I|]$}
		\State Query $z_i$, where $i$ is the $j$-th (smallest) element of $I$.
		\State Let $U \times V$ be a $z_i$-monochromatic rectangle of $g$ such that 
		\State \ \ (1) $A_{I\setminus\{i\}}^{i,U} \times B_{I\setminus\{i\}}^{i,V}$ is $\tau$-thick,
		\State \ \ (2) $\alpha_{I\setminus\{i\}}^{i,U} \geq  (1- 3\delta) \alpha$,		
		\State \ \ (3) $\beta_{I\setminus\{i\}}^{i,V} \geq \frac{1}{\varphi} (1- 3\delta)\beta$,
		\State Update $A = A^{i,U}, B = B^{i,V}$ and $I = I \setminus \{ i \}$.
	\EndIf
\EndWhile
\State Output $f \circ g^{\,p}(A \times B)$.
\end{algorithmic}
\end{algorithm}
\end{minipage}

\paragraph*{Correctness.} 
The algorithm maintains an invariant that $A_I \times B_I$ is $\tau$-thick. This invariant is trivially true at the beginning.

If both $A_I$ and $B_I$ are $\varphi$-average-thick, the algorithm finds sets $A'$ and $B'$ on line 5--7 as follows.
Consider the case that Alice communicates at node $v$. She is sending one bit. Let $A_0$ be inputs from $A$ on which Alice sends 0 at node $v$
and $A_1 = A \setminus A_0$. We can pick $c\in\ZO$ such that $|(A_c)_I| \ge |A_I|/2$.
Set $A''=A_i$. Since $A_I$ is $\varphi$-average-thick, $A''_I$ is $\varphi/2$-average-thick. So using Lemma \ref{LEM:THICK} on $A''_I$ with $\delta$ set to $1/2$, we can
find a subset $A'$ of $A''$ such that $A'_I$ is $\frac{\varphi}{4 \cdot |I|}$-thick and $|A'_I| \geq  |A''_I|/2$.  ($A' \subseteq A''$ will be 
the pre-image of $A'_I$ obtained from the lemma.)
Since $\varphi = 4 \cdot 2^{- \eps h}$ and $|I| \le p \le 2^{h(1 - \eps)}$, the set $A'_I$ will be $2^{-h}$-thick, i.e. $\tau$-thick. Setting $B'=B$, the rectangle $A' \times B'$ satisfies properties from lines 6--7. A similar argument holds when Bob communicates at node $v$.

If $A_I$ is not $\varphi$-average-thick, the existence of $U\times V$ at line 11 is guaranteed by Lemma \ref{LEM:PROJ}. Similarly in the case
when $B_I$ is not $\varphi$-average-thick.

Next we argue that the number of queries made by Algorithm \ref{ALG:SIMUL} is at most $5C/ \eps h$. In the first part of the \textbf{while} loop (line 3--8), the density of the current $A_I\times B_I$ drops by a factor $4$ in each iteration. There are at most $C$ such iterations, hence this density can drop by a factor of at most $4^{-C} = 2^{-2C}$. For each query that the algorithm makes, the density of the current $A_I \times B_I$ increases by a factor of at least $(1-3\delta)/\varphi \ge \frac{1}{2 \varphi} \ge 2^{\eps h - 3}$ (here we use the fact that $\delta \le 1/6$). Since the density can be at most one, the number of queries is upper bounded by
\[
  \frac{2 C}{\eps h - 3} \le \frac{4 C}{\eps h}, \tag*{when $h \ge 6 / \eps$.}
\]

Finally, we argue that $f(A \times B)$ at the termination of Algorithm \ref{ALG:SIMUL} is the correct output. Given an input $z \in \bool^p$, whenever the algorithm queries any $z_i$, the algorithm makes sure that all the input pairs $(x,y)$ in the rectangle $A \times B$ are such that $g(x_i, y_i) = z_i$ --- because $U\times V$ is always a $z_i$-monochromatic rectangle of $g$. At the termination of the algorithm, $I$ is the set of $i$ such that $z_i$ was not queried by the algorithm.  As $p > 4C/\eps h$, $I$ is non-empty. Since $A_I\times B_I$ is $\tau$-thick, it follows from Lemma \ref{LEM:REG} that $A\times B$ contains some input pair $(x,y)$ such that $g^{|I|}(x_I, y_I) = z_I$, and so $g^p(x, y) = z$. Since $\Pi$ is correct, it must follow that $f(z) = f \circ g^{\,p}(A\times B)$. This concludes the proof of correctness. \qed

\bsni
With greater care the same argument will allow for $\delta$ to be close to $\frac{1}{2}$. This would require also tightening the $1 - 3 \delta$ factors appearing in Lemma \ref{LEM:PROJ} to something close to $1 - 2 \delta$. The details are left to the reader, should he be interested.

\section{Hitting rectangle-distributions for \texorpdfstring{$\GH$}{GH}}
    \label{SEC:RGH}

We construct a hitting rectangle distribution for $\GH_{n, \frac{1}{4}}$. Subsequently, we will show a $(\delta, h)$-hitting rectangle distribution where  $\frac{|A \times B|}{|\zon \times \zon|} \geq 2^{-h}$.

Recall that $d_H(x,y)$ denotes the Hamming distance between the strings $x$ and $y$. Let $B_r(x)$ be the Hamming ball of radius $r$ around $x$, i.e. $B_r(x) = \{ y \in \ZO^n \mid d_H(x, y) \le r \}$; for a set $A \subset \ZO^n$, $B_r(A) = \cup_{a \in A} B_r(a)$. 

Let $\eps = \frac{1}{8}$ and $\cH$ be the set of all strings in $\zon$ with Hamming weight $n/2$. Now consider the rectangle distributions $\sigma_0$ and $\sigma_1$ obtained from the following sampling procedure:

\begin{itemize}
\item Choose a random string $x \in \cH$, and let $\bar x \in \cH$ be its bit-wise complement.

\item Now let $U_x = B_{\eps n}(x)$ and $V_x = B_{\eps n}(\bar x)$.

\item The output of $\sigma_1$ is the rectangle $U_x \times V_x$, and the output of $\sigma_0$ is $U_x \times U_x$.
\end{itemize}

\bsni
For the chosen value of $\eps$, $U_x \times V_x$ is a $1$-monochromatic rectangle, since for any $u \in U_x, v \in V_x$,
\[
 d_H( u, v ) \ge n - 2 \eps n \ge \frac 3 4 n.
\]
On the other hand, $U_x \times U_x$ is $0$-monochromatic, since for any $u, u' \in U_x$,
\[
d_H( u, u' ) \le 2 \eps n \le \frac 1 4 n.
\]
Both inequalities are obtained by a straight-forward application of triangle inequality.

\begin{lemma}\label{lem:rect-dist-gh}
The distributions $\sigma_0$ and $\sigma_1$ are $(2^{-\frac{n}{100}}, \frac{n}{100})$-hitting monochromatic rectangle distributions for $\GH_{n, \frac{1}{4}}$.
\end{lemma}

\bsni
To prove Lemma \ref{lem:rect-dist-gh}, we need the following theorem due to Harper. We will call $S \subset \zon$ a \textit{Hamming ball with center} $c \in \zon$ if $B_r(c) \subseteq S \subset B_{r+1}(c)$ for some non-negative integer $r$. For sets $S, T \subset \zon$, we define the \textit{distance} between $S$ and $T$ as $d(S,T) = \min\{d_H(s,t) \mid s \in S, t\in T\}$.

\begin{theorem}[Harper's theorem, \cite{FF81, H66}]
    \label{thm:harper}
Given any non-empty subsets $S$ and $T$ of $\zon$, there exist a Hamming ball $S_0$ with center $\bar{1}$ and Hamming ball $T_0$ with center $\bar{0}$ such that $|S| = |S_0|, |T|= |T_0|$ and $d(S_0, T_0)\geq d(S,T)$.
\end{theorem}

Note that Claim \ref{thm:harper} also tells us when $B_{r}(S)$ is smallest for a set $S \subset \zon$. This can be argued in the following way: Given a set $S \in \zon$, let $T_S = \zon \setminus B_r(S)$. It is immediate that $d(S,T_S) = r+1$. Now let us suppose that $S$ is such that it achieves the smallest $B_r(S')$ among all $S' \in \zon$ with $|S'| = |S|$. This also means that $T_S$ is the biggest such set. Using Harper's theorem, we can find set $S_0$ and $T_0$ such that $d(S_0,T_0) \geq r+1$ where $S_0$ is centered around $\bar{1}$ and $T_0$ is centered around  $\bar{0}$ with $|S_0| = |S|$ and $|T_0| = |T_S|$. Now it is easy to see that $T_0 \subseteq \zon \setminus B_{r}(S_0)$, i.e., $|T_S| = |T_0| \leq |T_{S_0}|$, which is a contradiction. This means that $|B_r(S)|$ will be the smallest if $S$ is a Hamming ball centered around $\bar{1}$. This gives us the following corollary.

\begin{corollary}
    \label{cor:harper}
    For any non-negative integer $r \in [n]$ and among the set $\mathsf{A} = \{A \subset \zon \mid |A| = k\}$ for any $k$, if $A$ is a Hamming ball centered around either $\bar{1}$ or $\bar{0}$, then $|B_r(A)| \leq |B_r(A')|$ for any $A' \in \mathsf{A}$.
\end{corollary}

\bsni
Now we state the proof of Lemma \ref{lem:rect-dist-gh}.
\begin{proof}[Proof of Lemma \ref{lem:rect-dist-gh}]
We will show that any set $A \subset \zon$ of size $|A| \ge 2^{\frac{99}{100} n}$ will be hit by $U_x$ with probability $\ge 1 - 2^{-\frac{n}{100}}$. The lemma now follows since $U_x$ and $V_x$ have the same marginal distribution.

\bsni
The event $U_x \cap A = \emptyset$ happens exactly when $x \notin B_{\eps n}(A)$:
\[
\Pr_x[U_x \cap A = \emptyset] = \Pr_x[x \notin B_{\eps n}(A)] \le \frac{2^n - |B_{\eps n}(A)|}{2^n}.
\]
From Corollary \ref{cor:harper} we know that $|B_{\eps n}(A)|$ is smallest when $A$ is itself a Hamming ball around $0$ of the same density as $A$. I.e., if $|B_{\gamma n}(0)| \le |A|$, then
\[
|B_{\eps n}(A)| \ge |B_{\eps n}( B_{\gamma n}(0) )| = |B_{(\gamma + \eps)n}(0)|.
\]
For $\gamma = \frac{1}{2} - \frac{\eps}{2} = \frac{1}{2} - \frac{1}{16}$, and since $H(\gamma) < \frac{99}{100}$, we have
\[
|B_{\gamma n}(0)| \le 2^{H(\gamma) n} \le 2^{\frac{99}{100} n} \le |A|.
\]
And so $|B_{\eps n}(A)| \ge |B_{(\gamma + \eps)n}(0)| = |B_{\frac{n}{2} + \frac{n}{16}}(0)| \ge 2^n - |B_{\frac{n}{2} - \frac{n}{16}}(1)| \ge 2^n - 2^{\frac{99}{100} n}$. It now follows
\[
\Pr_x[U_x \cap A = \emptyset] \le \frac{2^{\frac{99}{100} n}}{2^n} \le 2^{-\frac{n}{100}}. \qedhere
\]
\end{proof}

\section{Hitting rectangle-distributions for \texorpdfstring{$\IP$}{IP}} \label{SEC:IP-DET}

In this section, we will show that $\IP_n$ has $( 4 \cdot 2^{-n/20}, n/5)$-hitting monochromatic rectangle-distributions. This will show a deterministic simulation result when the inner function is $\inprod_n$, i.e.,
\begin{align*}
\cD^{cc}(f \circ \inprod_n^p) \geq \cD^{dt}(f) \cdot \Omega(n).
\end{align*}

\bsni
We will use the following well-known variant of Chebyshev's inequality:

\begin{proposition}[Second moment method]\label{exact-second-moment} Suppose that $X_i \in [0,1]$ and $X = \sum_i X_i$ are random variables. Suppose also that for all $i$ and $j$, $X_i$ and $X_j$ are \emph{anti-correlated}, in the sense that
  \[
  \E[X_i X_j] \le \E[X_i] \cdot \E[X_j].
  \]
  Then $X$ is well-concentrated around its mean, namely, for every $\eps$:
  \begin{equation*}
    \Pr[ X \in \mu (1 \pm \eps) ] \ge 1 - \frac{1}{\eps^2 \mu}.
  \end{equation*}
\end{proposition}


\bsni
All of the rectangle-distributions rely on the following fundamental anti-correlation property:

\begin{lemma}[Hitting probabilities of random subspaces]
  \label{LEM:rand-subspace} Let $0 \le d \le n$ be natural numbers. Fix any $v \not= w$ in $\bbF_2^n$, and pick a random subspace $V$ of dimension $d$. Then the probability that $v \in V$ is exactly
  \[
  p_v =
  \begin{cases}
    \frac{2^{d}-1}{2^n-1} & \tif v \not= 0\\
    1 & \tif v = 0.
  \end{cases}
  \]
  And the probability that both $v, w \in V$ is exactly
  \[
  p_{v,w} =
  \begin{cases}
    \binom{2^{d}-1}{2} \bigm/ \binom{2^n-1}{2} & \tif v, w \not= 0\\
    p_v & \tif w = 0,\text{ and}\\
    p_w & \tif v = 0.
  \end{cases}
  \]
  Hence it always holds that $p_{v, w} \le p_v p_w$.
\end{lemma}

\proof The case when $v$ or $w$ are $0$ is trivial. The value $p_v = \Pr[v \in V]$ for a random subspace $V$ of dimension $d$ equals $\Pr[M v = 0]$ for a random non-singular $(n - d) \times n$ matrix $M$, letting $V = \ker M$. For any $v\not=0, v'\not=0$, $M$ will have the same distribution as $M N$, where $N$ is some fixed linear bijection of $F_2^n$ mapping $v$ to $v'$; it then follows that $p_v = p_{v'}$ always. But then
\[
\sum_{v\not=0} p_v = \E\left[ \sum_{v\not=0} [v \in V]\right] = 2^{d}-1,
\]
and since all $p_v$'s are equal, then $p_v = \frac{2^{d}-1}{2^n-1}$.

Now let $p_{v,w} = \Pr[v \in V, w \in V]$. In the same way we can show that $p_{v, w} = p_{v', w'}$ for all two such pairs, since a linear bijection will exist mapping $v$ to $v'$ and $w$ to $w'$ (because every $v\neq w$ is linearly independent in $\bbF_2^n$). And now
\[
\sum_{v, w\not=0} p_{v, w} = \E\left[ \sum_{v, w\not=0} [v \in V] [w\in V]\right] = \binom{2^{d}-1}{2}.
\]
The value of $p_{v, w}$ is then as claimed. We conclude by estimating
\[
\frac{p_{v, w}}{p_v p_w} = \frac{\binom{2^{d}-1}{2}}{\binom{2^n-1}{2}} \cdot \frac{1}{p_v p_w} = \frac{2^{d}-2}{2^{d}-1} \cdot \frac{2^n-1}{2^n-2} < 1. \qed
\]

\bsni
It can now be shown that a random subspace of high dimension will hit a large set w.h.p.:

\begin{lemma}\label{LEM:subspace-fools-subset-1}
   Let $\eps < \frac{1}{2}$ be a positive real number, and consider a set $B \subseteq \bool^n$ of density $\beta = \frac{|B|}{2^n} \geq 2^{-(\frac{1}{2} - \eps)n}$. Pick $V$ to be a random linear subspace of $\bool^n$ of dimension $d$, where $d \geq (\frac{1}{2} - \frac \eps 4) n + 6$. Then
  \begin{align*}
    \Pr_V\left[\frac{|B \cap V|}{|V|} \in (1\pm 2^{-\frac\eps 4 n}) \cdot \beta\, \right] \geq 1- \frac{1}{4} \cdot 2^{-\frac\eps 4 n}.
  \end{align*}
\end{lemma}

\begin{proof}
  Let $b_1, \ldots, b_N$ be the elements of $B$, and define the random variables $X_i = [b_i \in V]$ and $X = |B \cap V| = \sum_i X_i$. The $\E[X_i]$ were computed in the proof of Lemma \ref{LEM:rand-subspace}, which gives us
\[
  \mu = \E[X] = \sum_i \E[X_i] = \begin{cases}\beta 2^n \frac{2^d-1}{2^n-1} \quad \text{if $\bar{0} \notin B$}\\
  \beta 2^n \frac{2^d-1}{2^n-1} + (1- \frac{2^d-1}{2^n-1}) \quad \text{otherwise.}
  \end{cases}
\]
Let's look at the case where $\bar{0} \not\in B$. We can estimate $\mu$ as follows:\footnote{Throughout the proof we will use the fact that $(1 \pm \delta)^2 \subseteq 1 \pm 3\cdot \delta$, and also that $1 \pm \delta \subseteq 1 \pm \delta'$ whenever $\delta \le \delta'$.}
\[
  \mu \;=\; \left(1+\frac{1}{2^n-1}\right)(1- 2^{-d})\beta |V| \;\in\; (1 \pm 2^{-(\frac{1}{2} - \frac \eps 2) n})^2 \beta |V| \;\subseteq\; \left(1 \pm \frac{1}{3} \cdot 2^{- \frac{\eps}{2} n}\right) \beta |V|.
\]
When $\bar 0 \in B$ we still have $\mu \in (1 \pm 2^{-\frac \eps 2 n})\beta |V|$ , because
$1- \frac{2^d-1}{2^n-1} \le 1 \ll \frac{1}{3} \cdot 2^{- \frac \eps 2 n}  \beta |V|$. So this holds in both cases.

\medskip

Lemma \ref{LEM:rand-subspace} also says that $\E[X_i X_j] \le \E[X_i] \E[X_j]$ for all $i \not= j$. And so by the second moment method (Lemma \ref{exact-second-moment}):
\begin{align*}
\Pr\left[ X \in \mu \left(1 \pm \delta\right) \right] &\ge 1 - \frac{1}{\delta^2 \mu} \\
\intertext{which means,}
\Pr\left[ X \in (1 \pm 2^{-\frac \eps 2 n}) (1\pm \delta) \beta |V| \right] & \ge 1 - \frac{1}{\delta^2 \cdot \beta \cdot 2^d \cdot (1- 2^{-\frac \eps 2 n})} \\
\intertext{Taking $\delta = \frac{1}{3} 2^{-\frac \eps 4 n}$, we get,}
\Pr\left[ X \in (1 \pm 2^{-\frac \eps 4 n}) \beta |V| \right] & \ge 1 - \frac{9}{2^{-\frac{\eps}{2} n} \cdot 2^{-(\frac{1}{2} - \eps) n} \cdot 64 \cdot 2^{(\frac{1}{2} - \frac\eps 4) n} } \ge 1 - \frac{1}{4} \cdot 2^{-\frac\eps 4 n}. \qedhere
\end{align*}
\end{proof}

\bsni
 We will show a similar result when we pick the set $V$ in the following manner: First we pick a uniformly random odd-Hamming weight vector $a \in \bool^n$, and then we pick $W$ to be a random subspace of dimension $d$ within $a^\perp$, where $d \geq (\frac{1}{2} - \frac{\eps}{4}) n + 6$; then $V = a + W$.
 
\begin{lemma}\label{LEM:subspace-fools-subset-1a}
   Consider a set $B \subseteq \bool^n$ of density $\beta = \frac{|B|}{2^n} \geq 2^{-(\frac{1}{2} - \eps)n}$. Pick $V$ as described above. Then
  \begin{align*}
    \Pr_V\left[\frac{|B \cap V|}{|V|} \in \beta(1\pm 2^{-\frac \eps 4 n})\right] \geq 1 - 2^{-\frac \eps 4 n}.
  \end{align*}
\end{lemma}


\begin{proof}
Let $B' = (B - a) \cap a^\perp$ and let $\beta' = \frac{|B'|}{|a^\perp|}$. A string $a \in \ZO^n$ is called \emph{good} when
\[
  \beta' \eqdef \frac{|(B - a) \cap  a^\bot|}{| a^\bot|} \in \beta \cdot (1 \pm 2^{-\frac \eps 4 n}).
\]
We will later show that if $a$ is a uniformly random string of odd Hamming weight, then
\begin{align}\label{EQ:eq1}
\Pr_a\left[ a \text{ is good} \right] \ge 1 - \frac 2 4 \cdot 2^{-\frac \eps 4 n}.\tag{$\ast$}
\end{align}
For every good $a$, Lemma \ref{LEM:subspace-fools-subset-1} gives us:
\begin{align*}
\Pr_W\left[\frac{|B' \cap W|}{|W|} \in \beta'(1\pm  2^{-\frac \eps 4 n}) \;\middle|\; a\right] \geq 1- \frac 1 4 \cdot 2^{-\frac \eps 4 n}.
\end{align*}
Our result then follows by Bayes' rule.

\medskip

To prove (\ref{EQ:eq1}), suppose that $a$ is chosen to be a uniformly random non-zero string (i.e. with either even or odd Hamming weight). Then $a^\bot$ is a uniformly random subspace of dimension $n - 1 \gg (\frac{1}{2} - \frac\eps4)n + 6$. Hence by Lemma \ref{LEM:subspace-fools-subset-1},
\[\label{eq:prob-a-perp}
  \Pr_a\left[\frac{|B \cap  a^\bot|}{|a^\bot|} \in \beta \cdot (1 \pm 2^{-\frac \eps 4 n})\right] \ge 1 - \frac 1 4 \cdot 2^{-\frac \eps 4 n}.\tag{$\ast\ast$}
\]
Now $|a^\perp| = 2^{n-1}$, so if $a^\|$ denotes the complement of $a^\perp$ (in $\ZO^n$), then $|a^\|| = 2^{n-1}$ also, and
\[
  \frac{|B \cap  a^\bot|}{|a^\bot|} \in \beta \cdot (1 \pm 2^{-\frac \eps 4 n}) \iff |B \cap  a^\bot| \in \frac{1}{2} |B| \cdot (1 \pm 2^{-\frac \eps 4 n})  \iff \frac{|B \cap  a^\||}{|a^\||} \in \beta \cdot (1 \pm 2^{-\frac \eps 4 n}).
\]
So (\ref{eq:prob-a-perp}) also holds with respect to the rightmost (equivalent) event. Since a uniformly random non-zero $a$ has odd Hamming weight with probability $> \frac{1}{2}$, it must then follow that if we pick a uniformly random $a$ with odd Hamming weight, then: 
\[
  \Pr_a\left[\frac{|B \cap  a^\||}{|a^\||} \in \beta \cdot (1 \pm 2^{-n/20})\right] \ge 1 - \frac 2 4 \cdot 2^{-\frac \eps 4 n}.
\]
Now notice that $|a^\|| = |a^\bot|$ and that for odd Hamming weight $a$, $B \cap a^\| = (B - a) \cap a^\bot$; this establishes (\ref{EQ:eq1}).
\end{proof}

\bsni
The lemmas above are the key to constructing rectangle-distributions for $\inprod$.

\begin{lemma}
	\label{LEM:IP-HITTING} For all $0 < \eps < 1/2$ and every sufficiently large $n$,
	$\inprod_n$ has $(2 \cdot 2^{-\frac \eps 4 n}, (\frac{1}{2} - \eps) n)$-hitting monochromatic rectangle-distributions.
\end{lemma}

\begin{proof}
We define the distributions $\sigma_0$ and $\sigma_1$ by the following sampling methods:
\begin{description}
\item[Sampling from $\sigma_0$:] We choose a uniformly-random $\frac{n}{2}$-dimensional subspaces $V$ of $\bbF_2^n$, and let $V^\bot$ be its orthogonal complement; output $V\times V^\bot$.

\item[Sampling from $\sigma_1$:] First we pick $a \in \bool^n$ uniformly at random conditioned on the fact that $a$ has odd Hamming weight; then we pick random subspace $W$ of dimension $(n-1)/2$ from $a^\bot$, and let $W^\bot$ be the orthogonal complement of $W$ \emph{inside} $a^\bot$. We output $V\times V^\|$, where $V = a + W$ and $V^\parallel = a + W^\bot$.
\end{description}

\noindent
The rectangles produced above are monochromatic as required. Also, $V$ and $V^\bot$ of $\sigma_0$ are both random subspaces of dimension $\ge (\frac{1}{2} - \frac\eps4) n + 6$ --- as required by Lemma \ref{LEM:subspace-fools-subset-1} --- and $V$ and $V^\|$ of $\sigma_1$ are both obtained by the the kind of procedure required in Lemma \ref{LEM:subspace-fools-subset-1a}. It then follows by a union bound that if $R$ is chosen by either $\sigma_0$ or $\sigma_1$ that, if $A, B$ are subsets of $\ZO^n$ of densities $\alpha, \beta \ge 2^{-(\frac{1}{2}-\eps)n}$, then
\[
  \Pr_{R}\left[ \frac{|A\times B \cap R|}{|R|} = (1 \pm 9 \cdot 2^{-\frac \eps 4 n}) \cdot \alpha \beta \right] \ge 1 - 2\cdot 2^{-\frac \eps 4 n}.
\]
Hence the same probability lower-bounds the event that $A\times B \cap R \neq \varnothing$.
\end{proof}

\section*{Acknowledgement}

Part of the research for this work was done at the Institut Henri Poincar\'e, as part of the workshop \emph{Nexus of Information and Computation Theories}.

The research leading to these results has received funding from the European Research Council under the European Union's Seventh Framework Programme (FP/2007-2013)/ERC Grant Agreement n. 616787. The first author is partially supported by a Ramanujan Fellowship of the DST, India and the last author is partially supported by a TCS fellowship.

The research leading to these results has received funding from the Foundation for Science and Tecnology (FCT), Portugal, grant number SFRH/BPD/116010/2016.

\clearpage
\stepcounter{section}
\addcontentsline{toc}{section}{\thesection. References}
\bibliography{fork}

\end{document}